\documentclass[a4paper,12pt]{article}
\usepackage{cleveref}
\usepackage{amssymb,amsmath,latexsym,amsthm} \usepackage{xcolor}
\usepackage{enumerate} \usepackage{dsfont}
\usepackage{geometry} \usepackage{graphicx}

\usepackage{stmaryrd} 
\usepackage{authblk}

\usepackage[showonlyrefs]{mathtools}

\usepackage [normalem]{ulem} %

\usepackage[latin1]{inputenc}

\usepackage{yfonts}
\newcommand{\SUS}{\subset}
\newcounter{margcount} 
  { \end{itemize} }
  { \end{enumerate} }

\newcommand{\qqquad}[0]{\qquad\qquad}
\newcommand{\qqqquad}[0]{\qqquad\qqquad}
\newcommand{\BS}[0]{\backslash}
\newcommand{\NT}[0]{\notag}

\makeatletter
\newcommand{\VEC}[2][r]{
  \gdef\@VORNE{1}
  \left(\hskip-\arraycolsep%
    \begin{array}{#1}\vekSp@lten{#2}\end{array}%
  \hskip-\arraycolsep\right)}

\newcommand{\VECMOD}[2][r]{
  \gdef\@VORNE{1}
  \left|\hskip-0.5\arraycolsep%
    \begin{array}{#1}\vekSp@lten{#2}\end{array}%
  \hskip-0.5\arraycolsep\right|}

\def\vekSp@lten#1{\xvekSp@lten#1;vekL@stLine;}
\def\vekL@stLine{vekL@stLine}
\def\xvekSp@lten#1;{\def\temp{#1}%
  \ifx\temp\vekL@stLine
  \else
    \ifnum\@VORNE=1\gdef\@VORNE{0}
    \else\@arraycr\fi%
    #1%
    \expandafter\xvekSp@lten
  \fi}
\makeatother

\newcommand{\EE}{{\mathcal E}}
\newcommand{\FF}{{\mathcal F}}

\newcommand{\TT}{{\mathcal T}}

\newcommand{\R}{\mathbb{R}}
\newcommand{\N}{\mathbb{N}}
\newcommand{\Z}{\mathbb{Z}}
\newcommand{\C}{\mathbb{C}}

\newcommand{\alp}{\alpha}
\newcommand{\bet}{\beta}

\newcommand{\lam}{\lambda}
\renewcommand{\phi}{\varphi}

\newcommand{\Lam}{\Lambda}

\newcommand{\nin}{\not\in}

\renewcommand{\Re}{\mathfrak{Re}}

\def\XXint#1#2#3{{\setbox0=\hbox{$#1{#2#3}{\int}$}
\vcenter{\hbox{$#2#3$}}\kern-.5\wd0}}

\newcommand{\upref}[2]{\hspace{-0.8ex}\stackrel{\eqref{#1}}{#2}} 

\newcommand{\lupref}[2]{\hspace{0ex} \stackrel{\eqref{#1}}{#2}} 

\usepackage{color}
\definecolor{verylightblue}{rgb}{0.9,0.95,1}
\definecolor{lightblue}{rgb}{0.7,0.7,1}

\definecolor{eqyellow}{rgb}{0.9375,0.8984,0.5469}
\definecolor{subeqyellow}{rgb}{1,0.9373,0.8353}

\definecolor{mygreen}{rgb}{0.3, 0.6, 0.3} 
\definecolor{verylightgreen}{rgb}{0.95, 0.95, 0.95} 
\definecolor{verydarkgreen}{rgb}{0, 0.5, 0}
\definecolor{darkgreen}{rgb}{0.85, 0.85, 0.85}  
\definecolor{mydarkgreen}{rgb}{0, 0.5, 0} 

\definecolor{mybrown}{rgb}{0.85, 0.4, 0.3}
\definecolor{verylightbrown}{rgb}{0.98, 0.72, 0.58}
\definecolor{verydarkbrown}{rgb}{0.44, 0.26, 0.26}

\definecolor{orange}{rgb}{1, 0.5, 0}

\definecolor{BurntOrange}{rgb}{1,0.356,0}

\definecolor{mydarkred}{rgb}{1,0.086,0.255}

\definecolor{RoseVYDP}{rgb}{0.84,0.086,0.255}

\definecolor{dgreen}{rgb}{0, 0.8, 0.5}     

\definecolor{CanaryBRT}{rgb}{1,0.76,0.26}

\definecolor{cyan}{rgb}{0, 1, 1}
\definecolor{verylightgray}{rgb}{0.95, 0.95, 0.95}
\definecolor{verylightgray}{rgb}{0.95, 0.95, 0.95}
\definecolor{verylightred}{rgb}{1, 0.8, 0.78}
\definecolor{verylightyellow}{rgb}{0.99, 0.98, 0.5}

\newcommand{\ignore}[1]{{}}

\newcommand{\NNN}[2]{\|#1\|_{#2}}

\newcommand{\skp}[2]{(#1, #2)}
\newcommand{\skpL}[3]{(#1, #2)_{#3}}

\newtheorem{theorem}{Theorem}[section]
\newtheorem{definition}[theorem]{Definition}
\newtheorem{corollary}[theorem]{Corollary}
\newtheorem{proposition}[theorem]{Proposition}
 \newtheorem{lemma}[theorem]{Lemma}
\theoremstyle{definition} \newtheorem{remark}[theorem]{Remark}

\DeclareMathOperator{\argmin}{argmin}

\newcommand{\digint}[2]{\ensuremath{\llbracket #1, #2 \rrbracket}}

\numberwithin{equation}{section}  \def\XXint#1#2#3{{\setbox0=\hbox{$#1{#2#3}{\int}$}
    \vcenter{\hbox{$#2#3$}}\kern-.5\wd0}} 
 \allowdisplaybreaks

\renewcommand{\em}[1]{\underline{#1}}

\usepackage{scalerel,stackengine}
\stackMath
\newcommand\widecheck[1]{%
\savestack{\tmpbox}{\stretchto{%
  \scaleto{%
    \scalerel*[\widthof{\ensuremath{#1}}]{\kern-.6pt\bigwedge\kern-.6pt}%
    {\rule[-\textheight/2]{1ex}{\textheight}}
  }{\textheight}%
}{0.5ex}}%
\stackon[1pt]{#1}{\scalebox{-1}{\tmpbox}}%
}

\makeatletter \newcommand\avsuminner[2]{%
  {\sbox0{$\m@th#1\sum$}%
    \vphantom{\usebox0}%
    \ooalign{%
      \hidewidth \smash{\vrule height\dimexpr\ht0+1pt\relax
        depth\dimexpr\dp0+1pt\relax}%
      \hidewidth\cr $\m@th#1\sum$\cr }%
  }%
} \makeatother

\author[*]{\rm Laurent B\'{e}termin} \author[**]{\rm Hans Kn\"upfer}
\affil[*]{\normalsize{QMATH, Department of Mathematical Sciences, University of Copenhagen, Universitetsparken 5, DK-2100 Copenhagen \O, Denmark. \texttt{betermin@math.ku.dk}. ORCID id: 0000-0003-4070-3344}}
\affil[**]{Institute of Applied Mathematics and IWR, University of Heidelberg, Im Neuenheimer Feld 205, 69120 Heidelberg, Germany. \texttt{knuepfer@uni-heidelberg.de} }

\begin{document}

\title{On Born's conjecture about optimal distribution of charges for an
  infinite ionic crystal} \date\today

\maketitle

\begin{abstract}
  We study the problem for the optimal charge distribution on the sites of a
  fixed Bravais lattice. In particular, we prove Born's conjecture about the
  optimality of the rock-salt alternate distribution of charges on a cubic
  lattice (and more generally on a $d$-dimensional orthorhombic
  lattice). Furthermore, we study this problem on the two-dimensional triangular
  lattice and we prove the optimality of a two-component honeycomb distribution
  of charges. The results holds for a class of completely monotone interaction
  potentials which includes Coulomb type interactions for $d\geq 3$. In a more
  general setting, we derive a connection between the optimal charge problem and
  a minimization problem for the translated lattice theta function.
\end{abstract}


\tableofcontents

\section{Introduction and setting}

\subsection{Introduction} \label{ss-intro} %

Ionic compounds are substances formed by charged ions, held together by
electrostatic forces. The ions are typically aligned in regular crystalline
structures, in an arrangement that minimizes the total interaction energy
between the positive ions ({\it cations}) and negative ions ({\it anions}).  A
large class of such materials are salts, formed by a reaction of an acid and a
base.  The material properties of these ionic compounds such as their high
melting point and their brittleness is determined by their specific lattice
structure and the distribution of charges within the lattice. A variety of
crystal lattice structures are observed as a function of the relative quantity
and size of the ions. While prediction of the expected lattice structure have
been made using contact number calculations \cite{Pauling}, in general, the
crystallization problem for ionic bounds has not been solved. Indeed, to
investigate stability of ionic lattice structures, the total interaction energy
between the particles needs to be calculated and compared for all possible
lattice configurations and distributions of the ions on these lattices. In
this paper, we consider the simpler question for optimizing the charge
distribution on a given crystal lattice. In particular, we prove Born's
conjecture about the optimal charge distribution for charges located at the
sites of a cubic or orthorhombic lattice.

\medskip

 The question asked by Born in \cite{Born-1921} (as recalled in
\cite{Latticesums}) is the following:
\begin{center}
  \begin{citation}
    ``\it How to arrange positive and negative charges on a simple cubic lattice
    of finite extent so that the electrostatic energy is minimal?'\footnote{``Ein
      endliches St�ck eines einfachen kubischen Raumgitters soll so mit gleich
      vielen positiven und negativen Ladungen von gleichem absoluten Betrage
      besetzt werden, da� die elektrostatische Energie des Systems m�glichst
      klein wird."}
  \end{citation}
\end{center}
His conjecture is that the alternate distribution of charges, i.e. when
$(-1)^{m+n+k}$ is the charge at the point $p=(m,n,k)\in \Z^3$, is the global
minimizer of the electrostatic energy among all the distributions of charges
with prescribed total charge. In \cite{Born-1921}, Born proved this conjecture
in dimension $d=1$ and he obtained the local minimality of the alternate
structure in dimension $d=3$. In this paper, we prove Born's conjecture in a
general setting of $d$-dimensional lattices and for a large class of interaction
energies. We further derive a connection of this problem to a minimization
problem for the translated lattice theta function associated with the dual of
the given lattice.

\begin{figure}[!h]
  \centering
  \includegraphics[width=8cm]{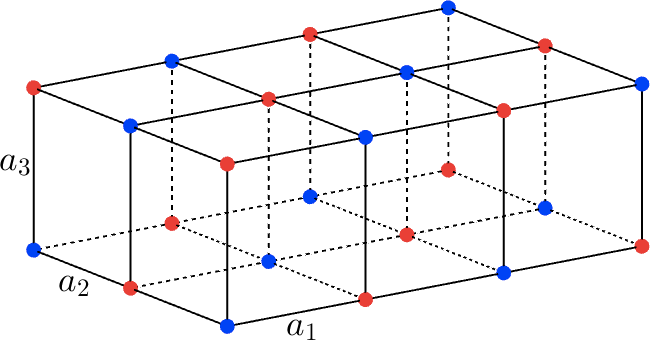}
  \caption{Optimal charge distribution for the orthorhombic lattice for $d = 3$. The points in blue (resp. red) have charge $-1$ (resp. $1$).}
  \label{fig-cubecharged}
\end{figure}

\medskip

We consider an ensemble of charges $\phi_x$ located at the vertices $x \in X$ of
a given $d$-dimensional Bravais lattice $X \subset \R^d$.  For technical
reasons, we also assume that the charges are $N$-periodic in each principal
lattice direction for some $N \in \N$ and can be represented by the finite
  sublattice $K_N$ (for the precise definitions, we refer to the next section).
The total interaction energy per lattice point can then be written as
\begin{align} \label{defenergyintro} %
  \EE_{X,f}[\phi] \ %
  = \frac{1}{2N^d} \sum_{y\in K_N} \sum_{x\in X\BS \{0\} } \phi_y
  \phi_{x+y}f(x),
\end{align}
for some radially symmetric interaction potential $f : \R^d \to [0,\infty)$. In
the case, when $f$ is not absolutely summable on $X \BS \{ 0 \}$, the classic
method of Ewald summation \cite{Ewald1} is used to give a meaning to the
infinite sum in \eqref{defenergyintro}. The class of interaction potentials $f$
we consider in particular includes all potentials $f(x) = F(|x|^2)$ for some
completely monotone function $F$. In particular, this includes all Riesz
potentials of the form $f(x)=|x|^{-s}$ for some $s > 0$. We consider the
minimization problem
\begin{align*}
   \phi \mapsto \EE_{X,f}[\phi]
\end{align*}
for all periodic charge
configurations satisfying a constraint for the total charge (see
  \eqref{totalcharge}) and for any given Bravais lattice $X\subset \R^d$.

\medskip

In our first result, we show that this minimization problem is related to a
minimization of the translated lattice theta function
\begin{align} \label{thet-1} %
 z\mapsto \theta_{X^*+z}(\alp) = \sum_{p\in X^*} e^{-\pi \alp |p+z|^2}, 
\end{align}
associated to the dual lattice $X^*$ in terms of the variable $z\in \R^d$ for
given $\alp > 0$ (see Theorem \ref{thm-1}).  At the core of our proof lies an
argument, originally due to Montgomery \cite{Mont} and generalized by one of the
authors in \cite{BeterminPetrache}. Theorem \ref{thm-1} can be used to calculate
the optimal charge distribution in specific Bravais lattices if the minimization
problem for the translated lattice theta function can be solved for these
lattices. We first consider the situation of a $d$-dimensional orthorhombic
lattice. In this case, the minimizer $z$ of \eqref{thet-1} is the center of the
primitive cell (for any $\alp$) \cite{BeterminPetrache}. For the triangular
lattice case, the minimizers are the two barycenters of the primitive triangles
forming the primitive rhombic cell \cite{Baernstein-1997}. In both case, the
knowledge of these minimizers gives us the minimal configuration of charges,
that are the alternate rock salt configuration in the orthorhombic case (Theorem
\ref{thm-2}), see Fig. \ref{fig-cubecharged}, and the honeycomb distribution for
the triangular lattice (Theorem \ref{thm-3}), see
Fig. \ref{fig-triangcharged}. In particular, Theorem \ref{thm-2} gives an
affirmative answer to Born's conjecture, cited above. Let us note that, in
  the case of Riesz potentials which are not summable over the lattice, also the
  analytic continuation of Epstein's zeta function has been used to describe the
  lattice energies in this case, see e.g. \cite{Emersleben}. Indeed, this
  approach yields the same energy as the Ewald method used in our approach.

\medskip

We note that the translated lattice theta function appears in several
mathematical models for physical systems with different kind of particles. For
example, Ho and Mueller \cite{Mueller:2002aa} wrote the interaction between two
Bose-Einstein condensates in terms of translated lattice theta functions, and
the same is done by Trizac et al. \cite{Samaj12,TrizacWigner16} in the context
of Wigner bilayers. Mathematically, the problem of minimizing
$z\mapsto \theta_{X+z}(\alp)$ was studied by Baernstein in
\cite{Baernstein-1997} for $X=\Lambda_1$ a two-dimensional triangular lattice
and by the first author in \cite{BeterminPetrache} for more general lattices.

\medskip

Let us also recall related work for optimal lattice configurations for systems
with the same kind of particles. The studies of lattice theta functions and
Epstein zeta functions are originally due to Krazer and Prym \cite{KrazerPrym}
and Epstein \cite{Epstein1}. Later, the problem of minimizing the Epstein zeta
function among two-dimensional Bravais lattices with a fixed density was studied
by Rankin \cite{Rankin}, Ennola \cite{Eno2}, Cassels \cite{Cassels} and Diananda
\cite{Diananda} (see also the recent review \cite{Henn}). They proved the
optimality of the triangular lattice (also called Abrikosov lattice in the
context of superconductivity). Montgomery \cite{Mont} proved the optimality of
the triangular lattice, among Bravais lattices with a fixed density, for the
lattice theta function $X\mapsto \theta_X(\alp)$, see also \cite[Appendix
A.2.]{NonnenVoros}. This result and its consequence \cite{CohnKumar} has been
used to show that the triangular configuration is a ground state of total
interaction energy for a class of interaction potentials. Let us cite the works
of Sandier and Serfaty \cite{Sandier_Serfaty} on Coulomb gases and
superconductivity and their consequences for the logarithmic energy on the
$2$-sphere in \cite{Betermin:2014rr}, but also the work of Aftalion, Blanc and
Nier \cite{AftBN} on Bose-Einstein Condensates and that of Nonnenmacher-Voros
\cite{NonnenVoros} on chaotic maps over a torus phase space. Furthermore,
Montgomery's result was used by the first author and Zhang in
\cite{Betermin:2014fy,BetTheta15} in order to prove the optimality of the
triangular lattice at high density for more general interaction
energies. Furthermore, the authors also used this result in order to prove the
optimality of the triangular lattice for radially symmetric, spatially extended
particles interacting via a radial potential in
\cite{BetKnuepfspatiallyextended}. In dimensions three, as recalled in
\cite[Section 2.5]{Blanc:2015yu}, the proof of the Sarnak-Str\"ombergsson
conjecture \cite{SarStromb} about the optimality, for the lattice theta
function, of the Body-Centered-Cubic (resp. Face-Centered-Cubic) lattice at high
(resp. low) density , would be an important advance both in analytic number
theory and in solid-state physics, see also \cite{SerfRoug15}. Some recent
advances have been made by the first author in
\cite{BeterminPetrache,Beterminlocal3d}, using recent results about Jacobi theta
functions in \cite{Faulhuber:2016aa}. In higher dimensions, the local minimality
of some lattices for the lattice theta function, the Epstein zeta function and
other related lattice energies was studied by Coulangeon et al. in
\cite{Coulangeon:kx,Coulangeon:2010uq,CoulLazzarini,coulschurm18}. The above investigation
are made under the assumption that the configuration can be expressed as a
Bravais lattice. In a different approach without periodicity assumptions optimal
lattice configurations have been studied e.g. in
\cite{VN1,VN2,Rad2,Rad3,Crystal,TheilFlatley,ELi,Stef1,Stef2,Luca:2016aa}.

\medskip

The optimality of the alternate configuration (also called ``chessboard
  configuration") in a different setting is also discussed in
  \cite{toricgrid}. Considering a flat torus $\TT$ of a certain specific size,
  composed by $N\in 2\N$ points, associated with an orthorhombic lattice, the
  authors ask the following question: How can $N/2$ points on this grid be
  located such that the associated energy
  $E=\sum_{p\neq q\in \TT} f(\delta(p,q))$ is minimized? Here, $f$ is a radial
  function and $\delta$ is the Lee distance, a graph distance counting the
  minimal number of edges of the grid connecting $p$ to $q$. Then, relaxing the
  problem by minimizing
  $\tilde{E}=\sum_{p\neq q\in \TT} f(\delta(p,q))w(p)w(q)$, where $w(p)$ is a
  weight (corresponding to our charges), and translating this problem in Fourier
  space, they prove their result for a specific choice of $f$ (including
  completely monotone functions) and $\TT$. It should be noted that there are
  some important difference in terms of the result and model in \cite{toricgrid}
  with respect to this paper, even though the strategy of the proof is similar.
  One difference is that the results in \cite{toricgrid} are concerned with
  finite sums in contrast to the infinite number of (long range) interactions
  energies considered in this work. Another difference is the graph distance
  considered in \cite{toricgrid} which requires combinatorically arguments. In
  particular, they numerically show (see \cite[Sect. 6.2]{toricgrid}) that the
  Lee distance cannot be replaced by the Euclidean distance for their result
  involving completely monotone functions. On the other hand, it doesn't seem to
  be straightforward to use our arguments to derive the results of
  \cite{toricgrid}. We also note that the results in this paper apply to any
  lattice for which the minimizer of the translated theta function is known.
\medskip

\textbf{Structure of the paper:} In the remaining parts of the first section, we
introduce the mathematical formulation for the model, introduce some needed
  special functions and present useful identities for those.
In Section \ref{sec-main}, we state our main
results in Theorems \ref{thm-1}--\ref{thm-5}. The proofs of these
theorems are given in Section \ref{sec-proofs}.

\medskip

\textbf{Notation:} We will write $(e_i)_{1\leq i\leq d}$ for the canonical basis
of $\R^d$. For any $x,y \in \R^d$, we denote the Euclidean scalar product by
$x\cdot y$.  We also use the notation $\digint{a}{b} := [a,b] \cap \Z$. 

We recall that a Bravais lattice in $\R^d$ is a set points of the form
  $X=\bigoplus_{i=1}^d \Z u_i \subset \R^d$ for a given set of linearly
  independent vectors $u_i \in \R^d$ with $i \in \digint{1}{d}$. We call $A_X$
  the generator matrix of the lattice $X$, i.e. $A_X \Z^d=X$. The associated
  quadratic form assigned with the Bravais lattice is given by
  $q_X(n) = \NNN{\sum_{i=1}^d n_i u_i}{}^2$ for $n \in \Z^d$.  The dual
    lattice $X^*$ of the lattice $X$ is given by
  \begin{align*}
    X^*=\left\{ p \in \R^d \ : \ p\cdot x\in \Z \text{ for all $x\in X$}
    \right\}.    
  \end{align*}

\subsection{The model} \label{ss-setting} %

We consider configurations of charges located on the sites of
  $d$-dimensional Bravais lattices $X \in \R^d$ for any $d \geq 1$.  We will
assume without loss of generality that all considered Bravais lattices have unit
density, i.e. the unit cell $Q:= \sum_{i=1}^d [0,1) u_i$ of the lattice has unit
volume. The general case can be recovered by rescaling the lattice.

\medskip

We consider charged lattices, where a charge is assigned to every lattice point:
\begin{definition}[Charged lattice] \label{def-lattice} %
  Let $d \geq 1$.
  \begin{enumerate}
  \item A charged lattice $L = (X,\phi)$ is a Bravais lattice $X = \bigoplus_{i=1}^d \Z u_i$ together with a
    function $\phi : X \to \R$ such that for any $x\in X$, the point $x$ has charge $\phi_x = \phi(x)$.
  \item We say that the charge distribution is $N$-periodic if it is periodic with period $N$ in any coordinate direction, i.e.
  \begin{align*}
    \phi(x+ N u_i) = \phi(x) %
    &&\text{for any $x \in X$ and any $i\in \digint{1}{d}$}.
  \end{align*}
  The charge distribution is called periodic if it is $N$-periodic for some
  $N \in \N$. We define the finite sublattice $K_N\subset X$ by
\begin{align}\label{def-KN}
  K_N   \ :=  \ \Big\{x=\sum_{i=1}^d m_i u_i\in  X \ : \  m_i\in \digint{0}{N-1}   \text{ for all } i \in \digint{1}{d} \Big\},
\end{align}
and we call $K_N^*$ the corresponding sublattice in $X^*$.
\item  The space of $N$-periodic functions on the lattice is denoted by
  $\Lam_N(X)$. It is equipped with inner product and norm by
\begin{align*}
  \skpL{\phi}{\psi}{K_N} = \sum_{y \in K_N} \phi(y) \overline{\psi(y)},&& \NNN{\phi}{} = \sqrt{(\phi,\phi)_{K_N}}\ .
\end{align*}
  \end{enumerate}
\end{definition}
Note that an $N$-periodic charge configuration is uniquely given by its values
on the $N^d$ points on the finite sublattice $K_N$.

\medskip

We will consider the following class of potentials:
\begin{definition}[The class of interaction potentials] \label{def-FF}%
  For $d\geq 1$, we say that $f\in \FF$ if $f:\R^d\to [0,+\infty)$ and if, for
  any $x\in \R^d \BS \{ 0 \}$, we have
  \begin{align} \label{FF-2} %
      f(x)=\int_0^{\infty} e^{-|x|^2t} d\mu_f(t),
    \end{align}
   where $\mu_f$ is a non-negative Borel measure.
If $f$ is absolutely summable over $X \BS \{ 0 \}$, we use the notation $f \in \ell^1(X \BS \{ 0 \})$.
\end{definition}
Note that the assumption in Definition \ref{def-FF} corresponds to the
particular case of G-type potential defined in \cite[Definition
1]{SaffLongRange} where $\mu_f$ is non-negative.
\begin{remark}[Relation to completely monotone functions] \label{rem-FF}
  \text{} %
  The class $\FF$ of admissible potentials is quite large. Indeed, by the Hausdorff-Bernstein-Widder theorem
  \cite{Bernstein-1929}, \eqref{FF-2} is equivalent to $f(x) = F(|x|^2)$ for
  some completely monotone function $F:(0,\infty)\to \R$, i.e. for any $F$
  which satisfies
  \begin{align*}
    (-1)^k F^{(k)}(r)\geq 0 && \forall r>0,\forall k \in \N.
  \end{align*}
  In particular, every
  potential of type
  \begin{align*}
    r^{-s}, \ e^{-\lam r^\alp}, \ \frac{e^{-\lam r}}r, \ \frac{e^{-\lam \sqrt{r}}}{\sqrt{r}} &&
                                                                                                  \text{ for $s >0, $ $\lam>0$, $\alp\in (0,2]$}, 
  \end{align*}
  is included in the class $\FF$.
  Further examples can be constructed by noting that if $f,g$ are completely
  monotone, then $\alp f+g$, $\alp > 0$ and $fg$ are completely monotone.  On
  the other hand, the class of Lennard-Jones-type potentials of the form
  $V(r)= r^{-p} - b r^{-q}$ with $b > 0$ and $p>q$ are not included in $\FF$.
  
\end{remark}
\begin{remark}[The logarithmic potential in $\R^2$] %
  In two dimensions, the Coulomb potential $f(x)=-\log |x|$ does not belong to
  $\FF$. However, we believe that the results in this paper should hold for
  this potential as well. Indeed, this should follow by an approximation
  argument together with the formula
  \begin{align*}
    -\log |x|=\frac{1}{2} \lim_{\varepsilon\to 0^{+}}\Big( \int_\varepsilon^{\infty} \frac{e^{-t |x|^2}}{t}dt + \gamma +\log \varepsilon   \Big),
  \end{align*}
  where $\gamma \approx 0.577$ is the Euler-Mascheroni constant.
\end{remark}
We assume that the interaction energy between two points of charges $\phi_x$,
$\phi_y$ at positions $x,y \in X$ is given by $\phi_x \phi_y f(x-y)$ for some
rotationally symmetric interaction potential $f \in \FF$. If $f$ is absolutely
summable over $X \BS \{ 0 \}$, the total potential energy is obtained directly
by summing the interaction energies for all points of the lattice. If $f$ is not
summable, the method of Ewald summation \cite{Ewald1} (introduced by Ewald in
his 1912 doctoral thesis) can be used to still define the total energy of the
system assuming that the total net charge is zero. More precisely, as Born did
in \cite{Born-1921} for the Coulomb potential, we use the classic Gaussian
convergence factors method, as in
\cite{deLeeuw:1980aa,Ewaldpolytropic,SaffLongRange}:
\begin{definition}[Interaction energy] \label{def-energy} %
  Let $N\geq 2$ and $L = (X,\phi)$ be a charged lattice with $N$-periodic charge
  distribution $\phi$ and let $f\in \FF$.  If $f \nin \ell^1(X \BS \{ 0 \})$, we
  assume the charge neutrality condition,
  \begin{align} \label{neutral} %
    \sum_{y\in K_N} \phi_y = 0.
  \end{align}
  The energy per particle is given by
  \begin{align} \label{ene} %
    \EE_{X,f}[\phi] \ %
    := \lim_{\eta \to 0} \Big(\frac{1}{2N^d} \sum_{y\in K_N} \sum_{x\in X\BS \{0\} } \phi_y \phi_{x+y}f(x) e^{-\eta
    |x|^2} \Big). %
  \end{align}
\end{definition}
We note that there are a variety of different ways to define the lattice energy
for non-integrable interaction potentials (see e.g.  \cite{Latticesums}). The
method of Ewald summation is commonly used to calculate the energy for different
charged systems and has been optimized for computational speed, see
e.g. \cite{deLeeuw:1980aa,PerramLeeuw,Ewaldpolytropic}. We also note that the
Ewald summation can also be used to calculate an energy for the case of
non-neutral charge configurations \cite{Ewaldpolytropic,AlastueyJancovici}. If
the interaction potential is given by a non-integrable Riesz potential,
i.e. $f(x) = |x|^{-s}$ for $s \in (0,d]$, then another way to define the energy
is by analytic extension of the Epstein zeta function, introduced in
\cite{Epstein1}. We give the definition of the Epstein zeta function in slightly
less generality as needed for our purposes.
\begin{definition}
  Let $X\subset \R^d$, $d \geq 1$, be a Bravais lattice associated with the
  positive definite quadratic form $q_X$. Then Epstein's zeta function is
  defined by
  \begin{align*}
     \textnormal{Z} \VECMOD{0;z}(q_X;s):=\sum_{n\in \Z^d\backslash \{0\}} \frac{e^{2i\pi n\cdot z}}{[q_X(n)]^{\frac s2}} && %
                                                                                                          \text{for $z\in \R^d$,  $s \in \{ \C : \Re s > d \}.$}
  \end{align*}
\end{definition}
This function admits an analytic continuation beyond its domain of absolute
convergence. For $s \in \C$ with $\Re s > 0$ and
$z\not\in X^*$, the extension is given by
\begin{align}\label{AnalyticEpstein}
  &\pi^{-\frac s2}\Gamma(\frac{s}{2})\textnormal{Z} \VECMOD{0;z} (q_X;s) \\
  &\qquad = -\frac{2}{s}+\int_1^{\infty}\sum_{n\in \Z^d\backslash \{0\}} e^{2i\pi n\cdot z} e^{-\pi t q_X(n)} t^{\frac s2-1}dt 
    +\int_1^{\infty}\sum_{n\in \Z^d} e^{-\pi t q_{X^*}(n+z) } t^{\frac {d-s}2-1}dt, \NT
\end{align}
where $q_{X^*}$ is the quadratic form associated with the dual lattice $X^*$,
see \cite{Epstein1}.

\medskip

\subsection{Discrete Fourier transform and convolution}

The proofs are formulated in terms of the discrete Fourier transform (for an
introduction, see e.g. \cite[Chapter 6]{DFT}) on the space of $N$-periodic
functions $\Lam_N(X)$ on the lattice $X$. We first note that an orthonormal
basis of $\Lam_N(X)$ is given by the functions , we define the functions by the
functions $e^{(k)} \in \Lam_N(X)$ where
\begin{align*}
  e^{(k)}(y)= \frac 1{N^{\frac d2}} e^{\frac{2\pi i}N y \cdot k}.
\end{align*}
The discrete Fourier transform is then defined as follows:
\begin{definition}[Discrete Fourier transform] %
  For any $\phi \in \Lam_N(X)$, its discrete Fourier transform $\widehat \phi \in \Lam_N(X^*)$ is given by
  \begin{align*}
    \widehat \phi(k) = \skpL{\phi}{e^{(k)}}{K_N} &&\text{for $k \in X^*$}.
  \end{align*}  
  where $e^{(k)}(y) := \frac 1{N^{\frac d2}} e^{\frac{2\pi i}N y \cdot k}$. For
  $\psi \in \Lam_N(X^*)$, the inverse Fourier transform is
  \begin{align*}
    \widecheck \psi(x) = \skpL{\psi}{\overline{e^{(x)}}}{K_N^*} &&\text{for $x \in X$}.
  \end{align*}
\end{definition}
Since the functions $e^{(k)}$ form an orthonormal basis, the Fourier transform
is a bijective map $\Lam_N(X) \to \Lam_N(X^*)$ whose inverse is given by the
inverse Fourier transform. Furthermore, Plancherel's identity holds with
constant $1$, i.e.
\begin{align*}
  \skpL{\phi}{\phi}{K_N} = \skpL{\widehat \phi}{\widehat \phi}{K_N^*},
\end{align*}
since the discrete Fourier transform corresponds to the application of a unitary
matrix.  A simple calculation shows that
$\widehat{\phi * \psi}(k) = \widehat \phi (k) \widehat \psi (k)$ for any
$k\in X^*$, where the convolution is defined by
\begin{align*}
  (\phi * \psi)(p) = \sum_{q \in K_N} \phi_q \psi_{p-q}.
\end{align*}

\subsection{Theta functions and useful identities} \label{ss-theta} %

Theta functions play an important role in different fields of mathematics. For
the computation of lattice sums, the following Jacobi theta function is useful:
\begin{definition}[Jacobi Theta functions] \label{def-theta} The third Jacobi theta
    function is defined by
\begin{align}\label{def-jacobitheta}
\vartheta_3(\xi; z) := \sum_{k\in \Z} e^{i\pi k^2 z + 2 i \pi k \xi},  \quad \Im (z)>0,\xi\in \C.
\end{align}
\end{definition}

We recall some useful identities for the Jacobi theta function restricted to the
upper imaginary axis. This restriction has been considered by Montgomery in
\cite{Mont} (he wrote $\theta(t,\beta) := \vartheta_3(\bet, it)$) in the
context of lattice sums:
\begin{lemma} \label{lem-theta} %
  Let $t > 0$ and let $\bet \in \R$. Then
  \begin{enumerate}
  \item\label{theta-ii} %
    $\displaystyle \vartheta_3(\beta; it)=\prod_{r=1}^{\infty} \left( 1-e^{-2\pi rt}
    \right)\left( 1+2e^{-(2r-1)\pi t}\cos(2\pi\beta)+e^{-2(2r-1)\pi t} \right).$
  \item $\displaystyle \vartheta_3(\beta; it) > 0$.
  \item The map $\beta\mapsto \vartheta_3(\beta; it)$ is $1$-periodic. Furthermore, we
    have
    \begin{align} \label{bet-12} %
     \vartheta_3\Big(\frac{1}{2}-\beta; it\Big)=\vartheta_3\Big(\frac{1}{2}+\beta; it\Big) %
      && \text{ for any
                                                             $\beta\in [0,\tfrac
                                                             12]$.}
    \end{align}
  \end{enumerate}
\end{lemma}
\begin{proof}
  The product formula \ref{theta-ii} is proved for example in \cite[Chapter 10,
  Theorem 1.3]{SteinC}. The positivity of $\vartheta_3$ follows by expressing
  the right hand side in the Jacobi product representation (i) of the theta
  function as
  \begin{align*}
    \vartheta_3(\beta; it)=\prod_{r=1}^{\infty} 
    \big( 1-e^{-2\pi rt}\big) \big(\sin^2 (2\pi \bet)) + \big[e^{-(2r-1)\pi t} + \cos (2 \pi \bet)\big]^2 \big) >0.
  \end{align*}
  The periodicity is a direct consequence of formula
  \eqref{def-jacobitheta}. The statement \eqref{bet-12} finally follows from
  \eqref{def-jacobitheta} with the change of variables $q=-k$.
\end{proof}
Finally, we introduce the translated lattice theta function (see
e.g. \cite[Section 2.3]{ConSloanPacking}) which is a particular case of
generalized lattice theta functions that have been studied by Krazer and Prym in
\cite{KrazerPrym}:
\begin{definition}[Translated lattice theta function] %
  \label{def-thetaL} %
  Let $d\geq 1$ and let $X = \bigoplus_{i=1}^d \Z u_i \subset \R^d$ be a Bravais lattice.  Then the theta function of
  the translated lattice $X+z$ or translated lattice theta function is defined by
  \begin{align*}%
    \theta_{X+z}(\alp) \ := \ \sum_{x\in X} e^{-\pi \alp |x+z|^2} &&\text{for any $\alp>0$ $z\in \R^d$.}
  \end{align*}
\end{definition}
The function $\theta_{X+z}(\alp)$ can be understood as follows: Consider a
matrix of points at the lattice points of $X$, carrying a unit charge. Suppose
that the charges induce an interaction potential $\exp(-\pi \alp |x|^2)$. Then
$\theta_{X+z}(\alp)$ describes the (Gaussian) interaction energy between $z$ and
$X$.
\begin{remark}[Translated lattice theta function and heat flow] %
  An interpretation of $\theta_{X+z}$ in terms of the heat flow is given by
  Baernstein in \cite{Baernstein-1997}. Let $P_X$ be the temperature at point $z$
  and at time $t$, if at time $t=0$ a heat source of unit strength is placed at
  each point of $X$, i.e. $P_X$ is defined, for any $z\in \R^d$, any Bravais
  lattice $X\subset \R^d$ and any $t>0$ as the solution of
\begin{align*}
  \left \{ %
  \begin{array}{ll}
    \partial_t P_X(z,t) = \Delta_z P_X \qqqquad &\text{for  $(z,t) \in \R^d \times (0,\infty)$} \vspace{0.3ex} \\
    P_X(z,0) = \sum_{p \in X} \delta_{p} \qqqquad &\text{for $z \in \R^d$},
  \end{array}
  \right.
\end{align*}
where $\delta_p$ is the Dirac measure at $p \in \R^d$. Then
\begin{align*}
  \theta_{X+z}(\alp) = \frac{1}{\alpha^{\frac{d}{2}}} P_X\left(z,\frac 1{4\pi \alp}\right) &&\text{for $z \in \R^d$, $\alp > 0$.}
\end{align*}
\end{remark}
We next recall some basic facts related to the translated lattice theta function, introduced in Definition \ref{def-thetaL}.  We
first recall Jacobi's Transformation Formula:
\begin{proposition}[Jacobi's Transformation Formula]  %
  For any Bravais lattice $X\subset \R^d$ of density one, any $\alp>0$ and any
  $z\in \R^d$, we have
  \begin{align}\label{prp-jacobi}
    \theta_{X+z}(\alp) = \sum_{x\in X} e^{-\alp\pi |x+z|^2}=\frac{1}{\alp^{\frac d2}}\sum_{p\in X^*} e^{2\pi i p\cdot z}e^{-\frac{\pi |p|^2}\alp}.
  \end{align}
\end{proposition}
\begin{proof}
  A proof (based on Poisson's summation formula) can be found e.g. in
  \cite[Theorem A]{SaffLongRange}. See also \cite{Bochnertheta} for a proof of a
  more general formula.
\end{proof}
For one-dimensional lattices the translated lattice theta function can be
expressed in terms of Jacobi theta function. Furthermore, we state other useful
identities related to scaling and periodicity of the translated lattice theta
function:
\begin{lemma} \label{lem-thetaL} %
  Let $d \geq 1$ and let $X \SUS \R^d$ be a Bravais lattice. Let
  $\alp > 0, \bet \in \R$, $z \in \R^d$ and let $\lam \neq 0$. Then
  \begin{enumerate}
  \item
    $\displaystyle \theta_{\Z+\beta}(\alp) = \frac{1}{\sqrt{\alp}} \vartheta_3
    (\beta;i\alpha^{-1})$ \ for $X = \Z$.
  \item
    $\displaystyle \theta_{\lam X + z}(\alp) = \theta_{X+\frac{z}{\lambda}}(\alp
    \lambda^2)$,
  \item the map $z \mapsto \theta_{X+z}(\alp)$ is periodic w.r.t.  the unit
      cell $Q = \sum_{i=1}^d [0,1)u_i$.
  \end{enumerate}
\end{lemma}
\begin{proof}
  In terms of $t := \frac 1\alp$, identity (i) is equivalent to
  \begin{align} \label{i-eq} \vartheta_3 (\beta;it)=\frac{1}{\sqrt{t}}\sum_{n\in
      \Z}e^{-\frac{\pi(n+\beta)^2}t}.
  \end{align}  
  In turn \eqref{i-eq} is a direct application of Jacobi's transformation
  formula \eqref{prp-jacobi} for $X=\Z$. The identity (ii) is easily
  obtained by
  \begin{align*}
    \theta_{\lam X + z}(\alp) =\sum_{x\in X} e^{-\pi \alp |\lam x + z|^2}=\sum_{x\in X} e^{-\pi \alp\lambda^2|x+\frac z\lambda|^2}=\theta_{X+\frac z {\lambda}}(\alp\lambda^2).
  \end{align*}
  Finally, the statement (iii) follows from the periodicity of $X$.
\end{proof}
Consequently, in view of (iii) for any fixed lattice $X$ and for given
$\alp>0$, the problem of minimizing $z\mapsto \theta_{X+z}(\alp)$ can be
restricted to $Q=\sum_{i=1}^d [0,1)u_i$.

\medskip

\begin{lemma}[Symmetry of the theta function]\label{lem-thetasymmetry} Let $X=\bigoplus_{i=1}^d \Z u_i$ be a Bravais lattice and $T(z)=\sum_{i=1}^d u_i -z$ be the symmetry with respect to the center $c=\frac{1}{2}\sum_{i=1}^d u_i$ of the primitive cell $Q=\sum_{i=1}^d [0,1)u_i$. Then, for any $z\in \R^d$ and any $\alpha>0$, we have
  \begin{align*}
    \theta_{X+z}(\alpha)=\theta_{X+T(z)}(\alpha)
  \end{align*}
\end{lemma}
\begin{proof}
  We have, for any $z\in \R^d$ and any $\alpha>0$,
  \begin{align*}
    \theta_{X+T(z)}(\alpha)=\sum_{x\in X} e^{-\pi \alpha |x+T(z)|^2}=\sum_{x\in X} e^{-\pi \alpha |x+\sum_{i=1}^d u_i-z|^2}=\sum_{x\in X} e^{-\pi \alpha |x-z|^2}=\theta_{X+z}(\alpha),
  \end{align*}
  by the periodicity of the translated theta function (Lemma
  \ref{lem-thetaL}(iii)) and the symmetry $-X=X$.
\end{proof}

\section{Statement of main results} \label{sec-main} %

We consider minimizer of the interaction energy among periodic charge configurations.  As in the paper by Born
\cite{Born-1921}, we assume that the charges of the points of $K_N$ satisfy a constraint on the total charge per
periodicity cell, i.e.
\begin{align}\label{totalcharge} %
  \frac 1{N^d} \sum_{y\in K_N} \phi_y^2 \ = \ 1, && \phi(0) >  0.
\end{align}
Our first result is a statement which connects the minimization of the
energy for certain lattice energies with the minimization problem among vectors
for the translated lattice theta function:
\begin{theorem}[Optimal charge distributions and theta function] \label{thm-1} %
  Let $d\geq 1$, let $X=\bigoplus_{i=1}^d \Z u_i \subset \R^d$ be a Bravais
  lattice and let $f \in \FF$.  Suppose that
    \begin{align}\label{prop-minz0}
z_0 \ \in  \ \argmin \Big\{  \theta_{X^*+z}(\alp) : z \in \sum_{i=1}^d \lam_i u_i^* : \forall i,\lam_i \in [0,1) \Big \} %
      &&\forall \alp > 0,
  \end{align}
  i.e. $z_0$ is an absolute minimizer of the translated theta function
  associated to the dual lattice $X^*$ for all $\alp > 0$. Furthermore, suppose
  that $z_0\in \frac{1}{N}X^*$ for some $N\in \N$. Then the energy
  $\mathcal{E}_{X,f}$ is minimized among all periodic charge configurations
  $\phi$ which satisfy \eqref{totalcharge} (and which satisfy \eqref{neutral} if
  $f$ is nonsummable) by
  \begin{align}\label{generalmin}
    \phi^*(x) =c \cos \left(2\pi x\cdot z_0 \right) &&
                                                       \text{for 
                                                       $x \in X$},
  \end{align}
  where the value of the constant $c$ is determined by
  \eqref{totalcharge}.

  Furthermore, if \eqref{prop-minz0} has at most two solutions, then the charge
  configuration \eqref{generalmin} is the unique minimizer of $\EE_{X,f}$, up to
  symmetries keeping $X$ invariant. Otherwise, if \eqref{prop-minz0} has more
  than two solutions then there are infinitely many minimizing charge
  configurations which are pairwise not related by symmetries keeping $X$
  invariant.
\end{theorem}
We notice that $z_0$ defined by \eqref{prop-minz0} is assumed to be
  a minimizer of $z\mapsto \theta_{X^*+z}(\alpha)$ for any $\alpha>0$. As we
  will see, that is the case if $X^*$ is an orthorhombic or triangular lattice,
  but it is not clear for which lattices this property remains true. In
  \cite[Thm 1.5]{BeterminPetrache}, it has been shown that a deep hole of $X$,
  i.e. a solution $c$ to $\max_{c\in \R^d}\min_{p\in X^*} |c-p|$ is an
  asymptotic minimizer for $z\mapsto \theta_{X^*+z}(\alpha)$ as
  $\alpha\to +\infty$. Furthermore, in \cite[Thm 1.6]{BeterminPetrache}, it has
  been proved that the asymptotic minimizers vary according to the structure of
  the lattice in dimension $2$, as it was also briefly explained in
  \cite[p. 232]{Baernstein-1997}. Therefore, it appears that lattices $X$ with
  sufficient asymmetry do not have the same asymptotic minimizer for any
  $\alpha>0$. In dimension $2$, according to these results, rectangular and
  rhombic lattices (the two generating vectors of $X$ have the same length and
  the minimal interior angle of the unit cell belongs to $[\pi/3,\pi/2]$) as
  well as non-rhombic lattices such that the second layer (from the origin) of
  $X$ has cardinality $2$ or $6$ seem to be the good candidates for keeping the
  same minimizers for any $\alpha>0$. We are not aware of a proof of such
  property.  

\medskip

The explicit form of the minimizers in \eqref{generalmin} implies in particular that the net charge for any minimizing configuration must be zero.
We have the following corollary:
\begin{corollary}[Charge neutrality] \label{cor-neutral} %
  Suppose that $f \in \ell^1(X \BS \{ 0 \})$. Then for any minimizer $\phi^*$ of
  $\EE_{X,f}$, the total net charge of the configuration is zero,
  \begin{align*}
    \sum_{y\in K_N} \phi^*(y)=0.
  \end{align*}
\end{corollary}
We note that charge neutrality is assumed if $f$ is not absolutely summable. The
corollary shows that charge neutrality also holds in the absolutely summable
case.

\medskip

Another interesting question is if there are lattices where there is no periodic
optimal charge configuration. Here, Theorem \ref{thm-1} could be useful to prove
such a non-existence result:
\begin{remark}[Non-existence] %
  From the proof of Theorem \ref{thm-1}, the following statement can be easily
  deduced: Let $d\geq 1$. Let $X=\bigoplus_{i=1}^d \Z u_i \subset \R^d$ be a
  Bravais lattice with dual lattice $X^*$.  Suppose that $z_0$ is an absolute
  minimizer of the translated theta function for all $\alp > 0$, i.e. $z_0$
  satisfies \eqref{prop-minz0}. Furthermore, suppose that
  \begin{align*}
    \argmin \Big\{  \theta_{X^*+z}(\alp) : z \in \R^d \Big\} \cap \left\{ \frac 1N X^* : N \in \N \right\} = \emptyset.
  \end{align*}
  Then the absolute minimum of the energy $\mathcal{E}_{X,f}$ is not obtained
  within the class of periodic charge configurations.
\end{remark}

\medskip

Together with a result obtained in \cite[Cor. 3.17]{BeterminPetrache} about
minimization of the translated lattice theta function for the orthorhombic lattice, we
obtain the optimal charge distribution for orthorhombic (and in particular
cubic) lattices. Indeed, for any orthorhombic lattice $X=\bigoplus_{i=1}^d \Z (b_i e_i) \subset \R^d$, the unique minimizer of $z\mapsto \theta_{X+z}(\alpha)$ is, for any $\alpha>0$, the center of the unit cell $z_0=\frac{1}{2}\left(b_1,...,b_d \right)$. We then get the following result:
\begin{theorem}[Optimal charge distribution for orthorhombic
  lattices] \label{thm-2} %
  Let $d\geq 1$, let $f\in \FF$ and let
  $X=\bigoplus_{i=1}^d \Z (a_i e_i) \subset \R^d$ be an orthorhombic lattice
  with $a_i > 0$ for $i \in \digint{1}{d}$.  Then the minimizer
  $\phi^*_{\rm orth} : X \to \R$ of $\EE_{X,f}$ among all periodic charge
  distributions $\phi : X \to \R$ satisfying the normalization constraint
  \eqref{totalcharge}, and the charge neutrality constraint \eqref{neutral} if $f$
  is nonsummable, is given by
  \begin{align*}
    \phi^*_{\rm orth}\Big(\sum_{i=1}^d m_i a_i e_i \Big)  =(-1)^{\sum_{i=1}^d m_i} &&\text{for all $m_i \in \N$, $i \in \digint{1}{d}$}.
  \end{align*}
 The minimizer is unique up to translations by a vector in $X$.
\end{theorem}
Theorem \ref{thm-2} applies in particular for the case when $f$ is the Coulomb
potential and $d=3$. Hence, Theorem \ref{thm-2} proves and generalizes Born's
conjecture about optimal charge distributions in \cite{Born-1921} for cubic
lattices and more general orthorhombic lattices. Moreover, we find again Born's
result \cite[Section 4]{Born-1921} in dimension $d=1$.

\medskip

From the proof of Theorem \ref{thm-2}, the following result for general lattices
can be deduced:
\begin{corollary}[The case of alternating charges] \label{cor-2} %
  Let $d\geq 1$, let $f\in \FF$ and let
  $X=\bigoplus_{i=1}^d \Z u_i \subset \R^d$ be a Bravais lattice. Furthermore,
  suppose that the equation \eqref{prop-minz0} only has a unique solution
  $z_0 \in \sum_{i=1}^d [0,1) u_i^*$. Then the unique, up to translations by a
  vector of $X$, periodic charge configuration $\phi$ which minimizes the energy
  $\EE_{X,f}$ is the alternating configuration
  \begin{align*}
    \phi^*_{\rm orth}\Big(\sum_{i=1}^d m_i u_i \Big)  =(-1)^{\sum_{i=1}^d m_i} &&\text{for all $m_i \in \N$, $i \in \digint{1}{d}$}.
  \end{align*}
\end{corollary}
In the general case of arbitrary Bravais lattices, the minimizer of the translated lattice theta function $z\mapsto \theta_{X^*+z}(\alp)$
is not known. However, in the special case of the triangular lattice in two dimensions, this theta function is well
understood -- its minimizers are the two barycenters of the primitive triangles -- and we obtain the optimal solution for the charge distribution problem. Recall that the triangular lattice of
unit density in two dimensions is given by the set
\begin{align} \label{trilattice} %
  \Lambda_1= \Z u_1\oplus \Z u_2, && %
                                     \text{where
                                     $u_1 = \sqrt{\tfrac{2}{\sqrt{3}}}
                                     \VEC{1;0}$, \ \ %
                                     $u_2 = \sqrt{\tfrac{2}{\sqrt{3}}}
                                     \VEC{1/2;\sqrt{3}/2}$.}
\end{align}
The situation here is slightly different than in the orthorhombic case:
\begin{theorem}[Optimal charge distribution for the triangular
  lattice] \label{thm-3} %
  Let $\Lam_1 \subset \R^2$ be the triangular lattice, defined in
  \eqref{trilattice} and let $f\in \FF$. Then the minimizer of
  $\EE_{\Lambda_1,f}$ among all periodic charge distributions
  $\phi : \Lam_1 \to \R $ which satisfy the normalization constraint
  \eqref{totalcharge}, and the charge neutrality constraint \eqref{neutral} if $f$
  is nonsummable, is
  \begin{align*}
    \phi^*_{\rm tri}(m u_1+n u_2) = \sqrt{2} \cos\Big(\frac{2\pi}{3} (m +n ) \Big) &&\text{for $m,n \in \Z$.}
  \end{align*}
  This minimizer is unique, up translations and rotations which keep $\Lam$ invariant.
\end{theorem}
We note that by definition the lattice points of the triangular lattice are $1$-periodic. The charges are, however,
not periodic with the same periodicity. Indeed, the minimizer $\phi^*_{\rm tri}$ is $3$-periodic. The optimal
charge distribution for the triangular lattice is sketched in Figure \ref{fig-triangcharged}. One can think of the
lattice as being realized by two different types of particles. The first (depicted in red) carries the charge
$\sqrt{2}$, while the other type particle (depicted in blue) carries the charge $-\sqrt{2}/2$. Since there are twice as
many particles of the second type than of the first type, the configuration is charge neutral.
\begin{figure}[!h]
  \centering
  \includegraphics[width=9cm]{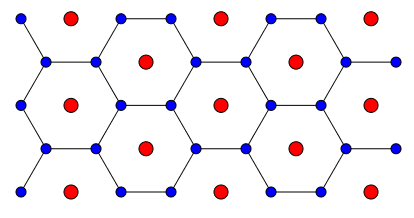}
  \caption{Optimal configuration $\phi^*_{\rm tri}$ in the case of triangular
    lattice}
  \label{fig-triangcharged}
\end{figure}

This honeycomb structure with two kind of particles given by Figure
\ref{fig-triangcharged} appears in different simulations and experiments. In
\cite{AssoudMessinaLowen}, this structure arises as the minimizer of a binary
mixture of particles interacting via the pair potential $V(r)=r^{-3}$ (up to a
multiplicative constant) of parallel dipoles, at zero temperature, when the
ratio of type particles is $(1/3,2/3)$ and at weak dipole-strength
asymmetry. This result explains the experimental finding of \cite[Figure
5.(a)]{Hay:2003aa}. Furthermore, the triangular lattice with the same absolute
values of charges (i.e. $\sqrt{2}$ and $\sqrt{2}/2$) as in Theorem \ref{thm-3}
is numerically identified in \cite{Xiao:1999aa,Levashov:2003aa} as a minimizer
of the Coulomb interaction energy when the particles, with positive charges, are
fixed on a triangular lattice embedded in a negative background charge to ensure
neutrality. However, contrary to our study, these works investigated the
minimizer of the interaction energy when the charges are fixed and as the
concentrations (or the dipole-strength asymmetry) of the species vary.

\medskip

Let us note that the method of Ewald summation in our case yields the same formula for the energy (up to an $s$
dependent constant) as the one obtained by analytic extension of Epstein's zeta function:
 \begin{theorem}[Relation to Epstein's Zeta function] \label{thm-5} %
   Let $d\geq 1$, $N\geq 2$, $X\subset \R^d$ be a Bravais lattice with generator matrix $A_X$, and
   $f_s(x)=|x|^{-s}$, $0< s\leq d$. For any $\phi$ satisfying \eqref{totalcharge}
   and \eqref{neutral}, for the $(\xi_k)_{k\in K_N^*}$ defined in
   \eqref{def-xi}, we have $\xi_k\geq 0$ for any $k$ and
\begin{align*} %
  \EE_{X,f_s}[\phi]=\frac{1}{2N^d}\sum_{k\in K_N^*} \xi_k \textnormal{Z} \VECMOD{0;A_X^{t}\frac{k}{N}} (q_X;s).
\end{align*}
\end{theorem}


\section{Proofs} \label{sec-proofs} %

\subsection{Proof of Theorem \ref{thm-1} and Corollary \ref{cor-neutral}}

We first give the proof of the Theorem in the case when the interaction
potential is absolutely summable over $X \BS \{ 0 \}$. We then will extend the
result to the case when the interaction potential is not absolutely summable.

\medskip

\paragraph{The absolutely summable case.} In the first part of the proof
(Lemma \ref{lem-s}--\ref{lem-reverse}), we follow the lines of the proof of Born
in a slightly more general setting to write the energy in terms of the Fourier
series \cite{Born-1921}, see Lemma \ref{lem-xi}. In the notation of Born, the
corresponding plane waves $e^{\frac{2i \pi}{N} x\cdot k}$ are called
``Grundpotentiale''. In the second step, we write the transformed expression of
the energy in terms of the translated lattice theta functions
$\theta_{X^*+\frac{k}{N}}(\alp)$.

\medskip

We first note express the energy in terms of the autocorrelation function
$s$:
\begin{lemma}[Energy expressed in $s$] \label{lem-s} %
  Let $d\geq 1$ and let $X\subset \R^d$ be a Bravais lattice. Let $N\geq 2$ and let $\phi : X \to \R$ be an $N$-periodic
  charge distribution satisfying \eqref{totalcharge}. Let $s : X \to \R$ be defined by
  \begin{align} \label{def-s} %
    s_x := \sum_{y\in K_N} \phi_y \phi_{y+x}.
  \end{align}
  Then $s \in \Lam_N(X)$, $s_{-x} = s_x$,
  $\displaystyle \sum_{x \in K_N} s_x = \Big(\sum_{x \in K_N} \phi_x\Big)^2$ and
  $s(0) = N^d$. The energy takes the form
  \begin{align} \label{ene-s} %
    \EE_{X,f}[\phi] \ %
    &\upref{ene}= \frac 1{2 N^d} \sum_{x\in X\BS \{0\} } s_x f(x).
  \end{align}
\end{lemma}
\begin{proof}
  The formula \eqref{ene-s} follows directly from the definition of energy
  \begin{align*}
    \EE_{X,f}[\phi] \ %
    \lupref{ene}= \frac{1}{2N^d} \sum_{y\in K_N} \sum_{x\in X\BS
      \{0\} } \phi_y \phi_{y+x}f(x). %
  \end{align*}
  by exchanging the sums.  The periodicity of $\phi$ implies $s_{-x} = s_x$ for any $x\in X$. One
  can easily check that $s(0) = \NNN{\phi}{K_N}^2=N^d$. Furthermore, by periodicity of $\phi$, we obtain
  \begin{align*}
  \sum_{x\in K_N} s_x=\sum_{x\in K_N} \sum_{y\in K_N} \phi_y \phi_{y+x}=\sum_{y\in K_N} \phi_y \sum_{x\in K_N} \phi_{x+y}=\Big(  \sum_{x\in K_N} \phi_x \Big)^2.
  \end{align*}
\end{proof}
It is convenient to express the energy in terms of the discrete (inverse)
Fourier transform $\xi$ of the autocorrelation function $s$. In the context of
signal processing, $\xi$ has also been called energy spectral density.
\begin{lemma} [Energy expressed in dual variables] \label{lem-xi} %
  Let $d\geq 1$ and let $X\subset \R^d$ be a Bravais lattice. For $N\geq 2$, let
  $\phi \in \Lam_N(X)$ satisfy \eqref{totalcharge} and let $s$ be defined as in
  \eqref{def-s}. Let $\xi \in \Lam_N(X^*)$ be given by
  \begin{align} \label{def-xi} %
    \xi_k:=\frac{1}{N^{d}}\sum_{y\in K_N} s_y e^{\frac{2\pi i}{N}y\cdot
    k} &&\text{for $k \in X^*$.}
  \end{align}
  Then $\xi_k \in \R$, $\xi_k \geq 0$,
  $\displaystyle \xi_0 =\frac{1}{N^d} \Big(\sum_{x \in K_N} \phi_x\Big)^2$ and
  $\xi_k=\xi_{-k}$ for all $k \in X^*$, and
  \begin{align} \label{xi-totalcharge} %
    \frac 1{N^{d}} \sum_{k \in K_N^*} \xi_k = 1.
  \end{align}
 Furthermore, the energy $E_{X,f}[\phi]$ can be equivalently written as
  \begin{align}\label{ene-xi} %
    \EE_{X,f}[\phi]=\frac{1}{2 N^{d}}\sum_{k\in K_N^*}\xi_k \sum_{x\in X\BS \{0\}} e^{\frac{2\pi i}N x \cdot k} f(x).
  \end{align}
\end{lemma}
\begin{proof}
  Note that $\xi = N^{-\frac d2}\widecheck s$. Since $s_{-x} = s_x$ for any
  $x\in X$, we have $\xi \in \R$.  The proof is based on Plancherel's identity on $\Lam_N(X)$. We give the calculation in
  detail below, using for notational simplicity, the convention $f(0) := 0$. We
  have, using the symmetry $f(-q)=f(q)$,
  \begin{align*}
    \sum_{x \in X \BS \{ 0 \}} s_x f(x) %
    &= \sum_{\ell \in X} \sum_{y \in K_N} s_y f(N\ell+y) %
    =\sum_{\ell \in X} \sum_{y \in K_N} \Big( \sum_{k\in K_N^*} \xi_k  e^{\frac{-2\pi i}N y \cdot k} \Big) f(N\ell+y)  \\
    &=  \sum_{k \in K_N^*}  \xi_k \sum_{x \in X \BS \{ 0 \}}  e^{\frac{2\pi i}N x\cdot k} f(x).
  \end{align*}
  Furthermore,
  $\xi = N^{-\frac d2}\widecheck s =N^{-\frac d2} \widecheck{\phi * (\phi \circ
    P)} = |\widecheck \phi|^2 \geq 0$, where $P(x) = -x$. The value of $\xi_0$
  follows from $\sum_{x\in K_N} s_x=\left( \sum_{x\in K_N} \phi_x \right)^2$
  (see Lemma \ref{lem-s}). Finally, one can easily check that
  $\skpL{\xi}{1}{K_N^*} = s(0) = N^{d}$.
\end{proof}
Note that due to the inversion symmetry of $s_x$ and $f(x)$, the sums on the
right hand side of \eqref{def-xi} and \eqref{ene-xi} are real numbers. By
combining the sums over positive and negative indices, \eqref{ene-xi} can
e.g. be written as
\begin{align*} %
  \EE_{X,f}[\phi]=\frac{1}{2N^{d}}\sum_{k\in K_N^*}\xi_k \sum_{x\in X\BS \{0\}} \cos\Big( \frac{2\pi}N x \cdot k \Big) f(x).
\end{align*}

\medskip

We also have the reverse transformation. Note that for given $\xi$, $\phi$ is not uniquely defined in general.
\begin{lemma} \label{lem-reverse} %
  Let $\xi \in \Lam_N(X^*)$ satisfy $\xi \geq 0$, $\xi_{-k} = \xi_k$ and
  \eqref{xi-totalcharge}.  Let
  \begin{align} \label{phi-reverse} %
    \phi_x=\frac{1}{N^{\frac d2}}\sum_{k\in K_N^*} \sqrt{\xi_k} \ \cos \Big(\frac{2 \pi}{N} x\cdot k \Big).
  \end{align}
  Then $\phi$ satisfies \eqref{totalcharge} and the formulas \eqref{def-s} and
  \eqref{def-xi} hold.
\end{lemma}
\begin{proof}
  We define $s \in \Lam_N(X)$ by
    $s = N^{\frac d2} \widehat \xi$, i.e.
    \begin{align*}
      s_x = \sum_{k \in K_N^*} \xi_k e^{-\frac{2\pi i}N x \cdot k},
    \end{align*}
    so that \eqref{def-xi} holds. A straightforward calculation now shows that
    the identities \eqref{totalcharge} and \eqref{def-s} are satisfied by the
    function $\phi$ defined in \eqref{phi-reverse}. 
\end{proof}
We turn to the proof of Theorem in the case when the interaction potential is
summable:
\begin{proposition}[Theorem $1$ --- the summable case] \label{prp-thm-11} %
  Suppose that the assumptions of Theorem \ref{thm-1} hold and suppose that
  $f \in \ell^1(X \BS \{ 0 \})$. Then the statement of Theorem \ref{thm-1} holds
\end{proposition}
\begin{proof}
  Let $s$ and $\xi$ be defined as in Lemma \ref{lem-s} and Lemma
  \ref{lem-xi}. Then we have
  \begin{align*}
    \EE_{X,f}[\phi]=\frac{1}{2N^{d}}\sum_{k \in K_N^*}\xi_k \sum_{x\in X\BS \{0\}} e^{\frac{2 \pi i}N x\cdot k} f(x).
  \end{align*}
  where $\xi \geq 0$ and $\xi$ satisfies \eqref{xi-totalcharge}. This suggests to
  consider
  \begin{align}\label{ene-xi4}
    E[k] \ :=\ \sum_{x\in X\BS \{0\}} e^{\frac{2 i\pi}N x\cdot k} f(x).
  \end{align}
  for $k \in X^*$.  We note that $E[k]$ cannot be minimized for $k \in NX^*$.
  Indeed, in this case the energy decreases by just switching the sign of a
  single charge in the periodicity cell. In the following, we hence assume that
  $k\in X^*\backslash NX^*$.

  \medskip

  In view of Definition \ref{def-FF} and by Fubini's Theorem, we get for any
  $k\in X^*\backslash NX^*$,
  \begin{align*}
    E[k] %
    &= \sum_{x\in X\backslash \{0\}} \int_0^{\infty}  e^{-|x|^2 t}e^{2\pi i x\cdot \frac{k}{N}}d\mu_f(t) %
      = \int_0^{\infty} \Big( \sum_{x\in X}e^{-|x|^2 t}e^{2\pi i x\cdot \frac{k}{N}}-1\Big)d\mu_f(t).
  \end{align*}
  By Jacobi's transformation formula \eqref{prp-jacobi}, this
  implies
  \begin{align*}
    E[k]&=  \int_0^{\infty} \Big(\pi^{\frac d2} t^{-\frac{d}{2}}\sum_{p\in X^*} e^{-\frac{\pi^2}{t}|p+\frac{k}{N}|^2}- 1\Big)  d\mu_f(t)\\
        &=\int_0^{\infty} \Big(\pi^{\frac d2} t^{-\frac{d}{2}} \theta_{X^*+ \frac{k}{N}}(\tfrac \pi t)  -1 \Big) d\mu_f(t).
  \end{align*}  
  Since $\mu_f$ is a non-negative measure, if $z_0$ is a minimizer on $X^*$ of
  $z\mapsto \theta_{X^*+z}(\alp)$ for all $\alp>0$, then $k_0=Nz_0$ minimizes
  $E$. By Lemma \ref{lem-reverse}, a minimizing
  configuration $\phi^*$ is then given, for any $x\in X$, by
  \begin{align} \label{theform} %
    \phi^*(x)= c \cos\big(2\pi x\cdot z_0 \big),
  \end{align}
  where $c$ is determined by the constraint \eqref{totalcharge}. This
    concludes the proof of existence.
  
  \medskip
  
  We turn to the proof of uniqueness. By assumption there are at most two
  minimizers $k_0$ and $k_1$ of $E[k]$. In view of Lemma
  \ref{lem-thetasymmetry}, these two minimizers are symmetry related, i.e. we
  have $\frac{k_1}{N} = \sum_{i=1}^d u_i^* - \frac{k_0}{N}$. Therefore,
  the minimizer of \eqref{ene-xi} in the class of functions $\xi^*$ which
  satisfy $\xi^* \geq 0$, $\xi_{-k}^*=\xi_k^*$ for any $k\in X^*$ and
  \eqref{xi-totalcharge} is hence given by $\xi^* \in \Lam_N(X^*)$, defined by
  $\xi_{k_0}^*=\xi_{k_1}^*=\frac{N^d}{2}$, by periodicity and the fact that
  $\xi_{k_0}^*=\xi_{-k_0}^*$, and $\xi_{k}^*=0$ for
  $k \in K_N^* \BS \{ k_0 ,k_1\}$. It follows that the corresponding
  autocorrelation function $s^*$ is given, for any $x\in X$, by
  \begin{equation} \label{uni-s} %
    s^*_x= \frac{N^d}{2}\left(e^{\frac{2i \pi}{N}k_0\cdot x}+e^{\frac{2i \pi}{N}k_1\cdot x}  \right)=\frac{N^d}{2}\left(e^{\frac{2i \pi}{N}k_0\cdot x}+e^{-\frac{2i \pi}{N}k_0\cdot x}  \right) 
  \end{equation}
  For any charge configuration $\phi^* \in \Lam_N(X)$, its (inverse) Fourier
    coefficients $\widecheck {\phi^*_k}$ satisfy the equation
    $|\widecheck{\phi^*_k}|^2 = \xi_k^*$ as a straightforward calculation
    shows. Any charge configuration with associated autocorrelation function
    $s$, given by \eqref{uni-s} therefore is of the form
    \begin{equation} \label{uni-s2} %
     \phi^*_x= \frac{\alp_1}{\sqrt{2}} e^{\frac{2i
            \pi}{N}k_0\cdot x}+ \frac{\alp_2 }{\sqrt{2}} e^{-\frac{2i \pi}{N}k_0\cdot x}
  \end{equation}
  for coefficients $\alp_i \in \C$ with $|\alp_i| = 1$. We next use the fact
  that the charge configuration $\phi^*$ is real. Furthermore, by a shift of
  coordinates, we can assume that $\phi^*_0 > 0$ and that $\phi^*$ attains its
  maximum at $x = 0$. With these assumptions, one can show that $\alp_1 = 1$ and
  $\alp_2 = 1$.  Therefore, in the case considered the charge distribution
  $\phi^*$ is uniquely determined by the autocorrelation function $s$ and is
  hence unique and given by \eqref{theform}.

\medskip

Suppose that there are at least three solutions of \eqref{prop-minz0}. Then as
above we can construct two symmetric functions
$\xi^{(1)}, \xi^{(2)} \in \Lam_N(X^*)$ satisfying the properties of Lemma
\ref{lem-xi}. Then any convex combination of
$\xi^{(\theta)} := \theta \xi^{(1)} + (1-\theta)\xi^{(2)}$ for
$\theta \in [0,1]$ yields a minimizing charge configuration $\phi^{(\theta)}$
(by Lemma \ref{lem-reverse}). We hence have constructed a one-parameter family
of optimal charge configurations $\phi^{(\theta)}$ such that the corresponding
autocorrelation functions $s_\theta$ are pairwise different.
\end{proof}

  \medskip
  
\begin{proof}[Proof of Corollary \ref{cor-neutral}] %
  Let $\phi^* \in \Lam_N(X)$ be a minimizer of $\EE_{X,f}$ and let $s^*$ and
    $\xi^*$ be defined by \eqref{def-s} and \eqref{def-xi}. In view of the proof
    of Theorem \ref{thm-1}, it then follows that $\xi^*$ is the convex
    combinations of functions $\xi$ with $\xi_0 = 0$. It follows that
    $\xi^*_0 = 0$. In view of Lemma \ref{lem-xi}, this shows that the
    configuration is charge neutral.
\end{proof}

\paragraph{The nonsummable case } 
We turn to the case when the potential energy is not summable. In this case, the
total energy of the lattice can be calculated using the Ewald summation method.
The idea of Ewald summation with Gaussian convergent factor is to approximate the energy by replacing the
interaction potential $f(x)$ by a family of screened interaction potentials
$f(x) e^{-\eta |x|^2}$ (for some small parameter $\eta > 0$), and to split the
screened interaction potential into a short-range part and a long-range part
(for some cut-off parameter $\alp > 0$). We follow the strategy in
  \cite{Ewaldpolytropic} where the Ewald summation has been used to calculate
  the energy of the Riesz potentials $f_s(x)=|x|^{-s}$ in dimensions
  $d = 1, 2,3$ and for $s \geq 1$ in a general setting without assuming charge
  neutrality.

\begin{theorem}[Ewald summation]\label{thm-riesz} %
  Let $d\geq 1$, let $X\subset \R^d$ be a Bravais lattice and $f\in \FF$. Then,
  for any $N\geq 2$, any $N$-periodic charge distribution $\phi$ satisfying
  \eqref{totalcharge} and \eqref{neutral}, and for any $\nu>0$, we have
  \begin{align*} \EE_{X,f}[\phi] %
    &= \frac{1}{2N^d}\sum_{x \in X \BS \{ 0 \}} s_x f_{1}^{(\nu)}(x) +
      \frac{1}{2N^d} \sum_{p\in X^* } \xi_{p}
      f_{2}^{(\nu)}\left(\frac{{p}}{N}\right)
      -\frac{\mu_f\left([0,\nu^2]  \right)}{2},
    \end{align*}
   where 
  \begin{align} \label{def-fs12} %
     f_1^{(\nu)}(x)=\int_{\nu^2}^{+\infty} e^{-t|x|^2}d\mu_f(t) , %
    && %
       f_{2}^{(\nu)}(x)=\pi^{\frac{d}{2}} \int_0^{\nu^2} t^{-\frac d2} e^{-\frac{\pi^2}{t}|x|^2}d\mu_f(t).
  \end{align}
\end{theorem}
\begin{proof}
We write the approximated interaction potential as
\begin{align}\label{alt-decompo}
f(x)e^{-\eta |x|^2}=\int_0^{\nu^2} e^{-(t+\eta)|x|^2}d\mu_f(t)+\int_{\nu^2}^{+\infty} e^{-(t+\eta)|x|^2}d\mu_f(t).
\end{align}
The second integral in \eqref{alt-decompo} is absolutely integrable and the limit $\eta\to 0$ can be taken directly. We hence obtain
\begin{align*}  %
  \EE_{X,f}[\phi] \ %
  &= \lim_{\eta \to 0} \Big(\frac{1}{2N^d} \sum_{x\in X\BS \{0\} } s_x  f(x) e^{-\eta |x|^2} \Big) %
  =  \frac{1}{2N^d} \sum_{x \in X \BS \{ 0 \}} s_x f_{1}^{(\nu)}(x) + I_2,
  \end{align*}
  where 
  \begin{align*}
I_2 := \lim_{\eta \to 0} \Big( \frac 1{2 N^d } \sum_{x \in X \BS \{ 0 \}} s_x \int_0^{\nu^2}  e^{-(t + \eta) |x|^2} \ d\mu_f(t)\Big).
  \end{align*}
For the term $I_2$, we transform into dual variables with help of Jacobi's transformation formula \eqref{prp-jacobi}. Taking into account the fact that $s_0 = N^d$, we have
  \begin{align*}
    \sum_{x \in X \BS \{ 0 \}} s_x e^{-(t + \eta) |x|^2} %
    &= \sum_{\ell \in X}\sum_{y\in K_N} s_y  e^{-(t + \eta) |y + N\ell|^2} - N^d\\
    &= \frac {\pi^{\frac d2}}{(t+\eta)^{\frac d2} N^d}  \sum_{y \in K_N} s_y \sum_{{p} \in X^*}   e^{2\pi i {p} \cdot \frac yN} e^{- \frac{\pi^2}{(t+\eta)N^2}|{p}|^2} - N^d.
    \end{align*}
    In view of the definition of $\xi$ in Lemma \ref{lem-xi}, we get
\begin{align*}
    \sum_{x \in X \BS \{ 0 \}} s_x e^{-(t + \eta) |x|^2} %
    &= \frac{\pi^{\frac d2}}{(t+\eta)^{\frac d2}} \sum_{{p} \in X^*}   \xi_{p} e^{- \frac{\pi^2}{(t+\eta)N^2}|{p}|^2} -N^d.
  \end{align*}
  Hence, since $\xi_0=0$ as a consequence of the charge neutrality of $\phi$, we
  arrive at 
  \begin{align*}
    I_2 &=\frac{\pi^{\frac{d}{2}}}{2N^d}\lim_{\eta\to 0} \sum_{p\in X^*} \xi_p \int_0^{\nu^2} (t+\eta)^{-\frac d2} e^{-\frac{\pi^2}{(t+\eta)N^2}|p|^2}d\mu_f(t)-\frac{1}{2}\int_0^{\nu^2} d\mu_f(t)\\
        &=\frac{\pi^{\frac{d}{2}}}{2N^d}\sum_{p\in X^*} \xi_p \int_0^{\nu^2} t^{-\frac d2} e^{-\frac{\pi^2}{t}\left|\frac{p}{N}\right|^2}d\mu_f(t)-\frac{\mu_f\left( [0,\nu^2] \right)}{2},
  \end{align*}
  and the result is proved.
    \end{proof}
    
Following the same arguments as in the summable case, we get
\begin{lemma} [Energy expressed in dual variables] \label{lem-xi-ns} %
  Let $d\geq 1$, $X\subset \R^d$ be a Bravais lattice and $f\in \FF$. Let $N\geq 2$ and
  let $\phi : X \to \R$ be a $N$-periodic charge distribution satisfying
  \eqref{totalcharge}. Let $s$ and $\xi$ be defined as in \eqref{def-s} and
  \eqref{def-xi}.  Then $\xi$ is $N$-periodic and satisfies $\xi \in \R$,
  $\xi \geq 0$ and $\xi_{-k} = \xi_k$. The constraint \eqref{totalcharge} takes
  the form $\skp{\xi}{1} = N^d$ and the energy $\EE_{X,f}[\phi]$ is expressed as
  \begin{align*}
    \EE_{X,f}[\phi]= &\frac{1}{2N^d} \sum_{k\in K_N^*} \xi_k %
    \Big(\sum_{x \in X \BS \{ 0 \}} e^{\frac{2\pi i}{N} x \cdot k} f_{1}^{(\nu)}(x) %
    + \sum_{q \in X^* } f_{2}^{(\nu)}(q + \tfrac kN) \Big) -\frac{\mu_f\left( [0,\nu^2] \right)}{2}.
  \end{align*}
\end{lemma}
\begin{proof}
  The argument for both sums proceeds analogously
  as in the proof of Lemma \ref{lem-xi}, using the fact that
  $\xi_0=N^{-d}\sum_{p\in K_N} s_p=0$ because $\phi$ satisfies
  \eqref{neutral}. 
\end{proof}
We are ready to give the proof of Theorem \ref{thm-1} in the non-summable case:
\begin{proof}[Proof of Theorem \ref{thm-1} in the non-summable case]
  We have already  shown that the statement of Theorem holds if $f$ is
    summable. We conclude the argument now for the non-summable case. By Lemma
  \ref{lem-xi-ns}, we have
  \begin{align} \label{renorm-xis} %
    \EE_{X,f}[\phi]= &\frac{1}{2N^d} \sum_{k\in K_N^*} \xi_k %
    \Big(\sum_{x \in X \BS \{ 0 \}} e^{\frac{2\pi i}{N} x \cdot k} f_{1}^{(\nu)}(x) %
    + \sum_{q \in X^* } f_{2}^{(\nu)}(q + \tfrac kN) \Big)  -\frac{\mu_f\left( [0,\nu^2] \right)}{2}.
  \end{align}
 Hence, it is enough to minimize
  \begin{align}\label{FkforEpstein}
    F[k] &:= %
    \sum_{x \in X \BS \{ 0 \}} e^{\frac{2\pi i}{N} x \cdot k} f_{1}^{(\nu)}(x)%
    + \sum_{q \in X^* } f_{2}^{(\nu)}(q + \tfrac kN)  \\
    &=\int_{\nu^2}^{+\infty} \sum_{x\in X\backslash \{0\}} e^{\frac{2\pi i}{N} x \cdot k} e^{-t|x|^2}d\mu_f(t)++ \pi^{\frac{d}{2}}\int_0^{\nu^2}\Big( \sum_{q \in X^*} e^{-\frac{\pi^2}{t}|q+\frac{k}{N}|^2} \Big)t^{-\frac{d}{2}}d\mu_f(t). \NT
  \end{align}
By application of Jacobi's transformation formula \eqref{prp-jacobi}, this implies
  \begin{align*}
    F[k] &= \int_{\nu^2}^{\infty}\Big(\frac{\pi^{\frac{d}{2}}}{t^{\frac d2}} \sum_{p\in X^*} e^{-\frac{\pi^2}{ t}|p+\frac{k}{N}|^2} -1  \Big)d\mu_f(t) \\
           &\qquad+ \pi^{\frac{d}{2}}\int_0^{\nu^2}\Big( \sum_{q \in X^*} e^{-\frac{\pi^2}{t}|q+\frac{k}{N}|^2} \Big)t^{-\frac{d}{2}}d\mu_f(t),
  \end{align*}
and we can rewrite this expression in terms of theta functions:
\begin{align*}
    F[k] &= \int_{\nu^2}^{\infty}\Big(\frac{\pi^{\frac{d}{2}}}{t^{\frac d2}} \theta_{X^*+\frac{k}{N}}\left( \frac{\pi}{t} \right) -1  \Big)d\mu_f(t) \\
           &\qquad+ \pi^{\frac{d}{2}}\int_0^{\nu^2} \theta_{X^*+\frac{k}{N}}\left( \frac{\pi}{t} \right)t^{-\frac{d}{2}} d\mu_f(t).
\end{align*}     
    We now conclude exactly as in the proof of the summable case.
 \end{proof}   

\subsection{Proof of Theorem \ref{thm-2} and Corollary \ref{cor-2}}

In this section, we give the proof of Theorem \ref{thm-2}. This proves Born's
conjecture for a general orthorhombic $d$-dimensional lattice distribution of
charges and for any interacting potential $f\in \FF$.

\medskip

Let us first recall the result the first author obtained in
\cite[Cor. 3.17]{BeterminPetrache} about the global optimality, for
$z\mapsto \theta_{X+z}(\alp)$, for any $\alp>0$, when $X$ is an orthorhombic
lattice. The result in \cite{BeterminPetrache} is based on a result by
Montgomery in \cite[Lemma 1]{Mont}. For the convenience of the reader, we
present the argument adapted to the particular case of our setting:
\begin{proposition}[\cite{BeterminPetrache}]\label{minsquare}
  Let $d\geq 1$ and $X=\bigoplus_{i=1}^d \Z (a_i e_i)$, where $a_i>0$ for any
  $1\leq i\leq d$.  Then we have
  \begin{align}\label{opticenter}
    \theta_{X+z^*}(\alp)\leq \theta_{X+z}(\alp), && %
                                                        \text{for any $\alp>0$ and any $z\in Q$,}
  \end{align}
  where the unique minimizer $z^*$ is the center of gravity of $Q$, i.e.
  \begin{align*}
    z^*=\frac{1}{2}(a_1,...,a_d). 
  \end{align*}
\end{proposition}
\begin{proof}
  For any $z \in \R^d$ and in view of the bijection $\Z^d \to X$,
  $n \mapsto \sum_{i} a_i n_i$, we have
  \begin{align*}
    \theta_{X+z}(\alp)=\sum_{x\in X} e^{-\pi\alp |x+z|^2} %
    =\sum_{n\in \Z^d} e^{-\pi \alp \sum_{i=1}^d (a_i n_i +z_i)^2} %
    =\prod_{i=1}^d \theta_{a_i\Z +z_i}(\alp)
    =\prod_{i=1}^d \theta_{\Z +\frac{z_i}{a_i}}(\alp a_i^2),
  \end{align*}
  where we have used Lemma \ref{lem-thetaL}(ii). We hence have reduced the
  problem to the minimization of the translated theta function on the
  one-dimensional lattice $\Z$ and it is enough to show that
  \begin{align} \label{oned-min} %
    \theta_{ \Z + \frac{1}{2}}(\alp)\leq \theta_{\Z +\frac{z_i}{a_i}}(\alp)
  \end{align}
  for any $\alp > 0$. In order to show \eqref{oned-min}, we use the identity
  \begin{align} \label{tran-clas} %
      \frac{1}{\sqrt{t}}\theta_{\Z+\beta}(t^{-1}) = \vartheta_3(\beta;it)
  \end{align}
of Lemma
  \ref{lem-thetaL}(i) which links the translated lattice theta function of the
  translated lattice to the Jacobi theta function. The Jacobi theta function can in turn be expressed in terms of the
  Jacobian product by
  \begin{align} \label{theta-prod} %
    \vartheta_3(\beta;it)=\prod_{r=1}^{\infty} \left( 1-e^{-2\pi rt} \right)\left(
    1+2e^{-(2r-1)\pi t}\cos(2\pi\beta)+e^{-2(2r-1)\pi t} \right),
  \end{align}
  see Lemma \ref{lem-theta}\ref{theta-ii}. From the representation
  \eqref{theta-prod}, it follows directly that, for any fixed $t>0$, the
  function $\beta \mapsto \vartheta_3(\beta;it)$ is decreasing on $[0,1/2]$ since each
  factor in \eqref{theta-prod} is positive and decreasing.  By Lemma
  \ref{lem-theta}(ii) and Lemma \ref{lem-theta}(iii), it hence follows that
  the theta function $\theta_{\Z+\beta}$ takes its minimum at $\beta\in [0,1/2]$. The
  estimate \eqref{opticenter} then follows in view of \eqref{tran-clas}. This
  completes the proof.
\end{proof}

\begin{proof}[Proof of Theorem 2]
  Let $d\geq 1$, $N\geq 2$ and $f\in \FF$. If $X=\bigoplus_{i=1}^d \Z u_i$, then
  $X^*=\bigoplus_{i=1}^d \Z u_i^*$ where $u_i^*=a_i^{-1}e_i$. By Proposition
  \ref{minsquare}, the unique minimum of $z\mapsto \theta_{X^*+z}(\alp)$ is
  \begin{align*}
    z_0=\frac{1}{2}\sum_{i=1}^d u_i^*,
  \end{align*}
  for all $\alp > 0$. Note that $z_0\in \frac{1}{N}X^*$ if only if $N \in
  2\N$.
  Therefore, the unique minimizer of \eqref{ene-xi} in the class of
    functions $\xi$ which satisfy $\xi \geq 0$, $\xi_{-k}=\xi_k$ for any
    $k\in X^*$ and \eqref{xi-totalcharge} is hence given by $\xi \in \Lam_N(X^*)$,
    defined by $\xi(k_0) = N^{d}$ and $\xi(k) = 0$ for
    $k \in K_N^* \BS \{ k_0 \}$. It follows that we get the autocorrelation
    function $s$ defined, for any $n\in \Z^d$ and any $x=\sum_{i=1}^d n_i u_i$,
    by
  \begin{equation*}
  s_x=N^d (-1)^{\sum_{i=1}^d n_i}.
  \end{equation*}
  Therefore, by Theorem \ref{thm-1}, we can uniquely
  reconstruct the charge distribution $\phi$ which is
  $\phi^*(\sum_{i=1}^d n_i u_i)= (-1)^{\sum_{i=1}^d n_i}$.
\end{proof}

  \begin{proof}[Proof of Corollary \ref{cor-2}] %
    By assumption, the set of points satisfying \eqref{prop-minz0} is given by a
    single point $z_0=\frac{k_0}{N}$ for some $k_0\in K_N^*$. In view of the
    proof of Theorem \ref{thm-1}, we infer that $z_0$ is the unique minimizer of
    the translated lattice theta function in
    $\sum_{i=1}^d [0,1) u_i \BS \{ 0 \}$. In turn, by Lemma
    \ref{lem-thetasymmetry}, it then follows that
    $z_0=\frac{1}{2}\sum_{i=1}^d u_i^*$ is the center of the unit cell of
    $X^*$. The argument is concluded as in the proof of Theorem \ref{thm-2}.
  \end{proof}

\subsection{Proof of Theorem \ref{thm-3}}\label{sec-triangle}

The proof for Theorem \ref{thm-3} follows by by using Theorem \ref{thm-1}
together with a result by Baernstein in \cite{Baernstein-1997} about the
minimizer for the translated theta function in the triangular lattice.  We first
note that the dual lattice of $\Lambda_1$ is the triangular lattice $\Lam_1^*$,
defined by
\begin{align*}%
  \Lambda_1^*= \Z u_1^*\oplus \Z u_2^*, && %
                                     \text{where
                                     $u_1^* = \sqrt{\tfrac{2}{\sqrt{3}}}
                                     \VEC{\sqrt{3}/2;-1/2}$, \ \ %
                                     $u_2^* = \sqrt{\tfrac{2}{\sqrt{3}}}
                                     \VEC{0;1}$,}
\end{align*}
i.e. $\Lambda_1^*$ and $\Lambda_1$ are the same lattice, up to
rotation.
\begin{figure}[!h]
   \centering
   \includegraphics[width=8cm]{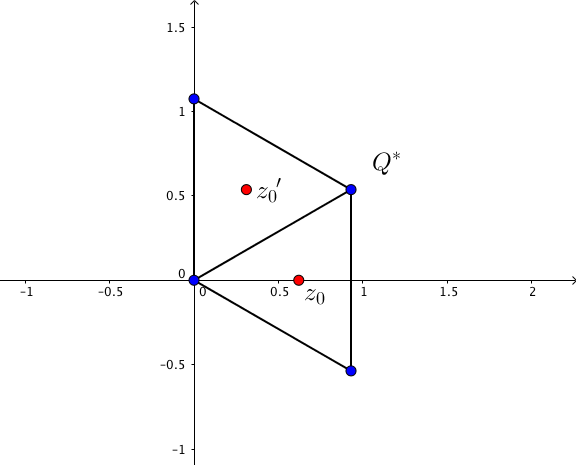}
   \caption{Primitive cell $Q^*$ of $\Lambda_1^*$ formed by two primitive triangle with barycenters $z_0$ and $z_0'$.}
   \label{Baernstein}
 \end{figure}
 For any $\alp>0$ $z\mapsto \theta_{\Lambda_1^*+z}(\alp)$ admits two minimizers
 in the set $Q^* :=[0,1)u_1^*+[0,1)u_2^*$, given by Baernstein's result
 \cite[Thm. 1]{Baernstein-1997}. These minimizers are the barycenters of the two
 primitive triangles forming $Q^*$ (see Fig.  \ref{Baernstein}), i.e.
\begin{align}\label{z0tri}
 z_0 = \frac 13 (u_1^* + u_2^*) \quad \textnormal{and}\quad
                                           z_0' = \frac 23 (u_1^* + u_2^*).
\end{align}
We note that $z_0$ and $z_0'$ belong to $\frac{1}{N}\Lambda_1^*$ if and only
  if $N \in 3\N$. Consequently, by Theorem \ref{thm-1}, the minimum among all
the periodic configurations is achieved for any $N\in 3\N$, by the
configurations defined, for any
$(m,n)\in \Z^2$, by 
\begin{align*}
  \phi^*_{\rm tri}(mu_1+nu_2) = c  \cos(2\pi (mu_1+nu_2)\cdot z_0)= c \cos\left(\frac{2\pi}{3}(m+n) \right).
\end{align*}
The value $c = \sqrt 2$ then follows from \eqref{totalcharge}.

\subsection{Proof of Theorem \ref{thm-5}}

We first show that
\begin{align} \label{fk-ep} %
F[k]=  \textnormal{Z} \VECMOD{0;A_X^t \frac{k}{N}}(q_X;s) +\frac{2\pi^{\frac{s}{2}}}{s\Gamma(\frac{s}{2})},
\end{align}
where $F[k]$ is given in \eqref{ene-xi4}. Using \eqref{AnalyticEpstein}, we
calculate
\begin{align}
  &\pi^{-\frac s2}\Gamma(\frac{s}{2})\textnormal{Z} \VECMOD{0;A_X^t \frac{k}{N}} (q_X;s) + \frac{2}{s} \\
  &\qquad = \int_1^{\infty}\sum_{n\in \Z^d\backslash \{0\}} e^{2i\pi n\cdot A_X^t \frac{k}{N}} e^{-\pi t q_X(n)} t^{\frac s2-1}dt 
    +\int_1^{\infty}\sum_{n\in \Z^d} e^{-\pi t q_{X^*}(n+A_x^t \frac{k}{N}) } t^{\frac {d-s}2-1}dt, \NT \\
  &\qquad = \int_1^{\infty}\sum_{x\in X \backslash \{0\}} e^{2i\pi x \cdot \frac{k}{N}} e^{-\pi t |x|^2} t^{\frac s2-1}dt 
    +\int_1^{\infty}\sum_{x \in X^*} e^{-\pi t |x+\frac{k}{N}|^2 } t^{\frac {d-s}2-1}dt, \NT
\end{align}
Now, we know (see e.g. \cite[Eq. (1.9)]{Laplacetransf}) that if
$f_s(x)=|x|^{-s}$, then
\begin{align*}
d\mu_{f_s}(t)=\frac{t^{\frac{s}{2}-1}}{\Gamma\left(\frac{s}{2}   \right)}  dt.
\end{align*}
Therefore, in view of \eqref{FkforEpstein} for $\alp=\sqrt{\pi}$,  we hence get \eqref{fk-ep} by a straightforward computation. Therefore, substituting $F[k]$ in \eqref{renorm-xis}, we finally obtain
\begin{align*}
\EE_{X,f_s}[\phi]&=\frac{1}{2N^d}\sum_{k\in K_N^*} \xi_k \textnormal{Z} \VECMOD{0;A_X^t \frac{k}{N}} (q_X;s)+\frac{1}{2N^d}\sum_{k\in K_N^*}\xi_k \frac{2\pi^{\frac{s}{2}}}{s\Gamma(\frac{s}{2})}-\frac{\pi^{s/2}}{s\Gamma(\frac{s}{2})}\\
&= \frac{1}{2N^d}\sum_{k\in K_N^*} \xi_k \textnormal{Z} \VECMOD{0;A_X^t \frac{k}{N}} (q_X;s)+\frac{\pi^{\frac{s}{2}}}{s\Gamma(\frac{s}{2})N^d}\left(\sum_{k\in K_N^*} \xi_k - N^d  \right)\\
&= \frac{1}{2N^d}\sum_{k\in K_N^*} \xi_k \textnormal{Z} \VECMOD{0;A_X^t \frac{k}{N}} (q_X;s)
\end{align*}
by \eqref{xi-totalcharge}.

\renewcommand{\em}[1]{\it{#1}}

\paragraph{Acknowledgements.} LB is grateful for the support of MATCH during his
stay in Heidelberg. Both authors would like to thank Florian Nolte for
interesting discussions.

\bibliographystyle{plain} \bibliography{Bornconject}
\end{document}